\documentclass[12pt]{elsarticle}
\usepackage[left=1in,right=1in, top=1.2in,bottom=1.2in]{geometry}
\usepackage{times}
\usepackage{eqnarray}
\usepackage{amssymb,amsthm,latexsym,amsmath,epsfig,pgf}

\newtheorem{theorem}{Theorem}[section]
\newtheorem*{theorem A}{Theorem A}
\newtheorem*{theorem B}{N\"olker's Theorem}
\newtheorem{lemma}{Lemma}[section]

\theoremstyle{remark}

\theoremstyle{remark}

\begin{document}

\begin{frontmatter}
\title{Rederiving the Upper Bound for Halving Edges using Cardano's Formula}



\author[label1]{Pintu Chauhan}
\author[label1]{Manjish Pal}
\author[label1]{Napendra Solanki}

\address[label1]{\small Department of Computer Science,\\
NIT-Meghalaya,\\
Shillong, India

\vspace*{2.5ex} 
 {\normalfont pintu.chauhan93@gmail.com, manjishster@gmail.com, napendrasinghrajput@gmail.com}
}

\begin{abstract}
In this paper we rederive an old upper bound on the number of halving edges present in the halving
graph of an arbitrary set of $n$ points in 2-dimensions which are placed in general position. We provide a different analysis
of an identity discovered by Andrejak et al, to rederive this upper bound of $O(n^{4/3})$. In the original paper of Andrejak et al.
the proof is based on a naive analysis whereas in this paper we obtain the same upper bound by tightening the analysis thereby
opening a new door to derive these upper bounds using the identity. Our analysis is based on a result of Cardano for finding the
roots of a cubic equation. We believe that our technique has the potential to derive improved bounds on the number of halving edges.

\let\thefootnote\relax\footnotetext{Received: xx xxxxx 20xx,\quad
  Accepted: xx xxxxx 20xx.\\[3ex]
  }
  
\end{abstract}


\end{frontmatter}

\section{Introduction}
Combinatorial Geometry is a field of discrete mathematics which deals with counting certain discrete 
structures in geometric spaces. There are several questions in this field which are very simple to state but finding 
the true solution is notoriously hard. One of the most important questions in this field is the celebrated $k$-set 
problem. Given a set of $n$ points in $\mathbb{R}^d$, a $k$-set is a subset of the given points which can be separated by
a $(d-1)$-dimensional hyperplane. The question is to find $f_d^k(n)$ which is the maximum number of such sets possible for an
arbitrary set of $n$ point in general position. This question remains to be elusive even in 2 dimensions, the status of which
is open for over forty years. Despite significant efforts put on this problem there is still a big gap between the upper and lower bounds on $f_2^k(n)$. In this paper we focus our study on $f_2^k(n)$. The special case of $f_2^{n/2}(n)$ is called the \emph{halving number problem}. From the perspective of finding upper bounds it is very easy to see that $f_2^k(n)$ is at most  $O(n^k)$. One can also observe that corresponding to every valid partition of the point set into $k$ and $n-k$ points there exists a unique line segment joining two points belonging to the set that divides the point sets into $k-1$ and $n-k-1$ points. This line segment is called a $k$-edge. The \emph{geometric graph} constructed as such is called the $k$-edge graph. This observation leads us to conclude that $f_2^k(n)$ is atmost the number of edges in this graph which can be atmost $O(n^2)$. For $k=n/2$, this graph is called the \emph{halving edge graph}. It is also well known that techniques to prove bounds for $f_2^{n/2}(n)$ can be translated to the general case of $f_2^{k}(n)$.\\

The first non trivial upper bounds for the halving case was obtained by 
Lov\'asz \cite{L71} and Erd\H{o}s, Lov\'asz, Simmons and Straus \cite{ELSS73} proved the upper bound of $O(n^{3/2})$
and a lower bound of $O(n \log n)$ in 1970's. The upper bound was later improved to $O(n^{3/2} \log^{*}n)$ by
Pach, Steiger and Szemer\'edi \cite{PSS92} in 1992. This was later improved to $O(n^{4/3})$ by Dey \cite{D98}. This bound was rederived by Andrejak, Aronov, Har-Peled, Seidel, Welzl \cite{AAHSW98} in a paper that also provided an identity connecting the number of halving edges and the number of \emph{crossings} in the halving edge graph in an arbitrary set of $n$ points in the plane. The lower bound was improved
by T\'oth to $ne^{\Omega(\sqrt{\log n})}$ \cite{T01}. This bound has been further refined recently by Nivasch to 
$\Omega(ne^{\sqrt{\ln 4}\sqrt{\ln n}}/ \sqrt{\ln n})$ \cite{N08}. From the lower bound perspective another line of approach has been 
considered by looking into $\gamma$-dense sets. These are sets in which the ratio of the maximum and minimum interpoint distance is atmost $\gamma \sqrt{n}$. For dense sets the upper bound can be improved to $O(n^{7/6})$ \cite{KT17}. For the $d$-dimensional case
the upper and lower bounds are known using complex techniques in combinatorics and algebraic geometry. It is known that $f_d^{n/2}(n) = O(n^{d-c_d})$ \cite{BFL90} where $c_d$ is a constant that goes to zero as $d$ increases. For lower bounds it is known that
given a construction of $\Omega(ng(n))$ halving edges in 2-dimension, one can extend this construction to 
$d$-dimensions, achieving a lower bound of $\Omega(n^{d-1}g(n))$ \cite{M02}. There are certain results which are obtained by looking
into special properties of the graph in three and four dimensions \cite{EVW97,S11,SST01,MSSW06}.

\section{Preliminaries}
In this paper we focus on the properties of the halving edge graph and the identity proven by Andrejak et al.
The \emph{halving graph} is graph defined by joining all the pair of points by segments whose corresponding lines partition
the point set into two equal parts i.e. consisting of $(n-2)/2$ points where $n$ is even. This
graph has several interesting properties most notably the \emph{Lov\'asz Local Lemma}. 

\subsection{Lov\'asz Local Lemma and Antipodality}
This result was proven by Lov\'asz in one of the early papers on this problem. Most of the upper
bound results on this problem use this lemma. The lemma basically says that if we consider any halving edge 
and rotate the line corresponding to that segment about one of its end points in both the directions (clockwise and anticlockwise) then the first point that this rotating line meets also form a halving edge when connected with
the point about which the rotation is performed. Although this result is a result that only 
reveals local properties of the point set, it is helpful in getting some global properties regarding the
halving edge graph. This technique was used by the proof of Dey to derive the bounds of $O(n^{4/3})$.

\subsection{Convex Chains and Dey's Original proof}
Dey in his original paper uses the notion of convex chains to derive the bound of $O(n^{4/3})$. The notion
of convex chains is based on using the Lov\'asz Lemma to derive an equivalence relation that results in the
formation of convex chains. In the paper, Dey proves that the number of convex chains constructed as such 
is atmost $O(n)$ and two convex chains cross at most twice with each other. Thus, creating at most $O(n^2)$ 
crossings which after using the Crossing Number Inequality again implies that the number of edges in the halving 
edge graph is atmost $O(n^{4/3})$.

\subsection{Main Identity}
In this section we describe the main identity that is proven by Andrejak et al. that
gives an exact relation between the number of crossings in the halving edge graph and 
the degree of nodes in the graph. This identity can be proven using doing an induction and
performing continuous motion on the set of points. 

\begin{eqnarray*}
\mbox{cr}(G) + \sum_{v} \binom{(d_v + 1)/2}{2} = \binom{n/2}{2}
\end{eqnarray*}

\section{Obtaining Dey's Bound using Cardano's Formula}
In this section, we prove some results regarding the halving edge graph which we use to get a different
analysis of the previous identity. Proofs of these are omitted as they are fairly easy to see. 

\begin{lemma}
Let $n_i$ be the number of nodes with degree $i$ in a halving edge graph $G$ and $d_v$ be the degree of a node $v$, then 
\begin{eqnarray*}
    \sum_v d_v^2 = \sum_i i^2 n_i
\end{eqnarray*}
\end{lemma}

\begin{lemma}
Let $n_i$ be the number of nodes with degree $i$ in a halving edge graph $G$ then $n_i$ is odd and $n_i \leq n_j \mbox{ } \forall i \geq j$. 
\end{lemma}

\begin{lemma}
Given an undirected graph $G$, if $d$ is the maximum degree of $G$ then $d \geq \frac{2m}{n}$. 
\end{lemma}

\subsection{Cardano's Formula}
Cardano's formula is a way to find the roots of a cubic polynomial. It was first published by Cardano in 1545. The formula says that given a cubic polynomial of the form
\begin{eqnarray*}
    ax^3 + bx^2 + cx + d = 0 
\end{eqnarray*}
The solution is given as
\begin{eqnarray*}
    x &=& \sqrt[3]{\left(\frac{-b^3}{27a^3} + \frac{bc}{6a^2} - \frac{d}{2a} \right) + \sqrt{\left( \frac{-b^3}{27a^3} + \frac{bc}{6a^2} - \frac{d}{2a} \right)^2 + \left( \frac{c}{3a} - \frac{b^2}{9a^2}\right)^3  }} \\
    &+&  \sqrt[3]{\left(\frac{-b^3}{27a^3} + \frac{bc}{6a^2} - \frac{d}{2a} \right) - \sqrt{\left( \frac{-b^3}{27a^3} + \frac{bc}{6a^2} - \frac{d}{2a} \right)^2 + \left( \frac{c}{3a} - \frac{b^2}{9a^2}\right)^3  }} - \frac{b}{3a}
\end{eqnarray*}
This formula can be simplified as 
\begin{eqnarray*}
x  &=& \{q + [q^2 + (r - p^2)^3]^{1/2}\}^{1/3} + {\{q - [q^2 + (r - p^2)^3]^{1/2}\}^{1/3}} + p
\end{eqnarray*}
where
\begin{eqnarray*}
p = \frac{-b}{3a}, q = p^3 + \frac{(bc-3ad)}{6a^2}, r = \frac{c}{3a}
\end{eqnarray*}

\begin{theorem}
The number of halving edges for an arbirary set of $n$ points is atmost $O(n^{4/3})$.
\end{theorem}
\begin{proof}
We prove this using a different analysis of the identity proven by Andrejak et al.
\begin{eqnarray*}
\mbox{cr}(G) + \sum_{v} \binom{(d_v + 1)/2}{2} = \binom{n/2}{2}
\end{eqnarray*}
In the paper by Andrejak et al, the second term in the identity which is a positive quantity
is ignored to be 0, there by implying the 
\begin{eqnarray*}
\mbox{cr}(G) \leq n^2/8 - n/8
\end{eqnarray*}
from which we can conclude that $\mbox{cr}(G) = O(n^{4/3})$ meeting the bound obtained by Dey.
In this proof we do not ignore the second and write the identity as an inequality  with variable $m$ where $m$ is the number of edges

\begin{eqnarray*}
\mbox{cr}(G) &+& \sum_v \left(\frac{d^2-1}{2}\right) = \frac{n^2}{8} - \frac{n}{8} \\
\mbox{cr}(G) &=& \frac{n^2}{8} - \frac{3n}{8} - \sum_v \frac{d^2}{2}  \\
\mbox{cr}(G) &=& \frac{n^2}{8} - \frac{3n}{8} - \sum_i \frac{i^2 n_i}{2}
\end{eqnarray*}
\end{proof}
Let $M =  \sum_i i^2 n_i$. Because of Lemma 4.2, $M \geq \sum_i^{\alpha} i^2 n_{\alpha} \geq \alpha^2 \geq 4m^2/n^2$ where $\alpha$ is the maximum degree of the halving edge graphs. Thus applying this to the identity we get

\begin{eqnarray*}
\mbox{cr}(G) &=& \frac{n^2}{8} - \frac{3n}{8} - \sum_i \frac{i^2 n_i}{2} \\
\mbox{cr}(G) &\leq& \frac{n^2-3n}{8} - \alpha^2/2 \\
&\leq & \frac{n^2-3n}{8} - \frac{2m^2}{n^2}
\end{eqnarray*}
Thus using the fact that cr($G$) $\geq \frac{m^3}{64n^2}$, we conclude that

\begin{eqnarray*}
\frac{m^3}{64n^2} \leq  \frac{n^2}{8} - \frac{3n}{8} - \frac{2m^2}{n^2} \implies
\frac{m^3}{64n^2} + \frac{2m^2}{n^2} + \left(\frac{n^2-3n}{8}\right) \leq 0 
\end{eqnarray*}

Thus we get a cubic inequality $P(m) \leq 0$ with $m$ as the variable and the cubic polynomial $P(m)$ as
$\frac{m^3}{64n^2} + \frac{2m^2}{n^2} + \left(\frac{n^2-3n}{8}\right)$. Now if we find the roots of this polynomial using
the Cardano's formula then we get a solution for the value of $m$ which is the number of halving edges in the 
halving graph. Applying the Cardano's formula we get $p = \frac{-b}{3a} = -\frac{128}{3}$, $q = p^3 - \frac{d}{2a} = -\frac{8\cdot(64)^3}{27} - \frac{64n^4 - 3n^3}{16}$ and $r = 0$. Thus we get the root of the polynomial $P(m)$ as

\begin{eqnarray*}
m &=& \Bigg\{-\frac{8\cdot(64)^3}{27} - \frac{64n^4 - 3n^3}{16} + \Bigg[\left(-\frac{8\cdot(64)^3}{27} - \frac{64n^4 - 3n^3}{16}\right)^2 - \frac{128^6}{3^6} \Bigg]^{1/2}\Bigg\}^{1/3} + \\
& & \Bigg\{-\frac{8\cdot(64)^3}{27} - \frac{64n^4 - 3n^3}{16} - \Bigg[\left(-\frac{8\cdot(64)^3}{27} - \frac{64n^4 - 3n^3}{16}\right)^2 + \frac{128^6}{3^6}\Bigg]^{1/2}\Bigg\}^{1/3} - \frac{128}{3} 
\end{eqnarray*}

Therefore, if we take the asymptotically dominating terms in the expression for $m$, we get $m = O(n^{4/3})$, reobtaining the upper bound of Dey.

\section{Discussion on the Bound of $M$}
In this section we discuss the implications of the bound of $M$ in the analysis of the above
identity. According to our analysis we have used a very crude lower bound on the value of $M$.
in the following we discuss how one can achieve better upper bounds for the number of halving edges using better lower bounds for $M$.
\begin{itemize}
    \item[(a)] From lemma 3.2 we observe that $i$ can only be an odd number. Thus a better lower 
    bound for $M$ could be obtained if we get a better lower bound using $\sum_{i = \mbox{odd}} i^2$ although
    we can't use the direct formula of $\sum_{i = \mbox{odd}} i^2$ because some of the terms in this series
    might be missing. A better bound on the number of missing terms i.e. number of $i$ such that $i$ is odd,
    can be used to get a better lower bound on $M$.
    
    \item[(b)] A better bound for $M$ can be obtained if we can get a reasonable estimate on $n_i$. In the
    current bound we have used a very naive lower bound for $n_i$ as 1 whereas an estimate for $n_i$ can
    give us a better bound.
\end{itemize}

\section{Conclusion and Future Work} 
in this paper we have provided a new analysis of bounding the number of halving edges in 2-Dimensions, this problems although we old not much understanding of of the exact bound of halving edges is present. We present a proof that rederived an old upper bound of $O(n^{4/3})$ our proof is based on a different analysis of the identity proven by Andrejak et al. our analysis is interesting because it does not neglect the term in the identity and we use a result of Cardeno's to obtained the aforementioned upper bound. This technique also allows us to apply more rigorous bounds on the neglected term, which may lead to more complex polynomials equations whose roots might be used. As part of future work we would like to investigate into deriving more stringent bounds on that term.

\end{document}